\documentclass[aps,pra,reprint,superscriptaddress,nofootinbib,longbibliography]{revtex4-2}
\usepackage{graphicx}
\usepackage{mathpazo}
\usepackage{dcolumn}
\usepackage{xcolor}
\usepackage{bm}
\usepackage{braket}
\usepackage{lipsum}
\usepackage{amsmath, amsthm, amsfonts, amssymb, bbm, dsfont, quantikz, ragged2e} 
\usepackage{enumitem}
\usepackage{mathtools}
\usepackage{nicematrix}
\usepackage{pifont}
\usepackage[most]{tcolorbox}
\usepackage[colorlinks,linkcolor=blue,urlcolor=blue, citecolor=blue]{hyperref}
\usepackage{tikz}

\usetikzlibrary{decorations, arrows, automata}

\usepackage{float}
\makeatletter
\let\newfloat\newfloat@ltx
\makeatother
\usepackage{algorithm}
\usepackage{algpseudocode}

\usepackage{tikz}
\usetikzlibrary{quantikz2}

\newtheorem{theorem}{Result}

\newcommand{\dyad}[1]{\vert #1 \rangle \langle #1 \vert}

\newcommand{\ben}[1]{{\color{black}#1}}

\newcommand{\expval}[1]{\langle #1 \rangle} 

\newcommand{\Paulichannel}[0]{\mathcal{E}_{\mathcal{P}}^{\alpha}}

\definecolor{teal}{RGB}{42, 157, 143}
\definecolor{yellow}{RGB}{233, 196, 106}
\definecolor{red}{RGB}{210, 66, 66}
\definecolor{lred}{RGB}{222,94,100}
\definecolor{lblue}{RGB}{179, 235, 242 }

\newcommand{\fref}[1]{\textcolor{blue}{\hyperref[#1]{Fig.$\,$\bfseries\ref{#1}}}}
\newcommand{\lemref}[1]{\textcolor{blue}{\hyperref[#1]{Lemma$\,$\bfseries\ref{#1}}}}
\newcommand{\thmref}[1]{\textcolor{blue}{\hyperref[#1]{Thm.$\,$\bfseries\ref{#1}}}}

\tikzset{
    every node/.style={font=\small},
    arrow/.style={-{Stealth}, thick},
    implies/.style={->, double equal sign distance, thick}
}

\newtcolorbox[auto counter]{mybox}[2][]{
    breakable = false,
    enhanced,
    sharp corners,
    colback=violet!3!white,
    colframe=violet!40!white,
    fonttitle=\bfseries,
    title={\centering \strut #2}, 
    enlarge bottom at break by=5mm,
    enlarge top at break by=5mm,
    overlay first={%
        \draw[black, line width=0.5mm](frame.south west)--(frame.south east);
        \node[anchor=north east] at (frame.south east) {continued on next page};
    },
    overlay middle={%
        \draw[black, line width=0.5mm](frame.south west)--(frame.south east);
        \draw[black, line width=0.5mm](frame.north west)--(frame.north east);
        \node[anchor=north east] at (frame.south east) {continued on next page};
        \node[anchor=south west] at (frame.north west) {continued from next page};
    },
    overlay last={%
        \draw[black, line width=0.5mm](frame.north west)--(frame.north east);
        \node[anchor=south west] at (frame.north west) {continued from last page};},
    #1
}

\renewcommand{\addcontentsline}[3]{}

\begin{document}

\preprint{APS/123-QED}

\title{An Algorithm for Estimating $\alpha$-Stabilizer Rényi Entropies via Purity}

\author{Benjamin Stratton}
\email{ben.stratton@bristol.ac.uk}
\affiliation{Quantum Engineering Centre for Doctoral Training, H. H. Wills Physics Laboratory and Department of Electrical \& Electronic Engineering, University of Bristol, BS8 1FD, UK}
\affiliation{H.H. Wills Physics Laboratory, University of Bristol, Tyndall Avenue, Bristol, BS8 1TL, UK}

\date{\today}

\begin{abstract}

Non-stabilizerness, or magic, is a resource for universal quantum computation in most fault-tolerant architectures; access to states with non-stabilizerness allows for non-classically simulable quantum computation to be performed. Quantifying this resource for unknown states is therefore essential to assessing their utility in quantum computation. The Stabilizer Rényi Entropies have emerged as a leading tool for achieving this, having already enabled one efficient algorithm for measuring non-stabilizerness. In addition, the Stabilizer Rényi Entropies have proven useful in developing connections between non-stabilizerness and other quantum phenomena. In this work, we introduce an alternative algorithm for measuring the Stabilizer Rényi Entropies of an unknown quantum state. Firstly, we show the existence of a state, produced from the action of a channel on $\alpha$ copies of some \ben{state $\rho$}, that encodes the $\alpha$-Stabilizer Rényi Entropy of $\rho$ into its purity. We detail several methods of applying this channel, and then, by employing existing purity-measuring algorithms, provide an algorithm for measuring the $\alpha$-Stabilizer Rényi Entropies for all integers $\alpha>1$. This algorithm is benchmarked for qubits and the resource requirements compared to other known algorithms. Finally, a non-stabilizerness/entanglement relationship is shown to exist in the algorithm, demonstrating a novel relationship between the two resources, \ben{before an instance of resource hiding is found}. 

\end{abstract}

\maketitle
\section{Introduction}

Not all quantum computations are hard. Given access to only stabilizer states, Clifford unitaries, and computational basis measurements, all implementable operations --- the so called \textit{stabilizer operations} --- can be simulated efficiently on a classical computer \cite{gottesman1997stabilizer, aaronson2004improved, Veitch_2014}. Thankfully for the field of quantum computation, these operations are also non-universal. However, the set of stabilizer operations are practically significant, as Clifford unitaries can typically be implemented easily (transversely) in fault-tolerant quantum computing frameworks \cite{campbell2010bound}, and stabilizer states are essential to most error correction protocols \cite{gottesman1997stabilizer}. Therefore, whilst stabilizer operations alone are not able to realise quantum advantage, they are operationally favourable. 

To perform non-classically-simulable quantum operations, and ultimately achieve universal quantum computation, one must be able to perform non-stabilizer operations. Fortunately, these can be achieved via stabilizer operations if they are supplemented with non-stabilizer states \cite{bravyi2005universal}; however, non-stabilizer states cannot be generated from stabilizer operations alone \cite{Veitch_2014}. Such a setting aligns with that of a quantum resource theory \cite{chitambar2019quantum}, which studies the limitations imposed by restricting operations to only an \textit{allowed} set. Those states that can be generated from just the allowed operations are termed the \textit{free} states, with all other states being considered resourceful. Here, stabilizer operations are regarded as the allowed operations and the stabilizer states as the free states \cite{Veitch_2014}. Access to non-free --- or resourceful --- non-stabilizer states then enables operations outside of the set of allowed operations to be implemented. In the case of qubits, if one has access to magic states \cite{bravyi2005universal}, they are able to perform universal quantum computation using stabilizer operations.

The non-stabilizerness (or magic) of a state is therefore considered to be a resource due to its ability to implement non-classically simulable quantum computation. Knowing the `amount' of non-stabilizerness a state possesses is then vital in understanding its usefulness in facilitating quantum advantage. In general, resources can be quantified via \textit{resource monotones} --- functions from states to the positive real-numbers that are non-increasing under allowed operations.

In the resource theory of non-stabilizerness, the $\alpha$ Stabilizer Rényi Entropies (SREs) \cite{PhysRevLett.128.050402} have emerged as a prominent resource monotone due to their computability \cite{passarelli2024nonstabilizerness, PhysRevB.107.035148, oliviero2022magic, chen2024magic, lami2023nonstabilizerness, tarabunga2024nonstabilizerness, haug2025efficientwitnessingtestingmagic}, and their ability to relate non-stabilizerness to other physical properties, such as information scrambling \cite{ahmadi2024quantifying} and entanglement \cite{gu2024pseudomagic, gu2024magic, PhysRevA.109.L040401}, to name a few. These are a family of resource monotones, with each value of \hbox{$0 < \alpha < \infty, \alpha \neq 1$} giving a separate monotone.  

The prominence of the SRE's as a measure of non-stabilizerness is furthered by the existence of algorithms to measure it for unknown quantum states \cite{PhysRevLett.132.240602, Oliviero2022}. These algorithms can outperform state tomography, potentially allowing for the non-stabilizerness of large scale systems to be benchmarked. 

In this work, we introduce a novel algorithm for measuring the SRE's of an unknown quantum state by developing a connection between the $\alpha$-SRE and purity. Specifically, we define a channel on $\alpha$ copies of some unknown pure state $\ket{\psi}$, such that the $\alpha$-SRE of $\ket{\psi}$ is encoded into the purity of the output state for all integers $\alpha > 1$. We give multiple methods for applying this channel and, by pairing with existing methods for measuring the purity from the literature, then detail an algorithm for measuring the $\alpha$-SREs for integer values of $\alpha$. Whilst the algorithm is shown to outperform quantum state tomography for even $\alpha$, it is not as resource efficient as other known algorithms \cite{PhysRevLett.132.240602}. However, it introduces the idea of estimating a resource monotone by encoding it into the purity of a state. This novel approach to the direct estimation of resource monotones could then be leveraged to develop new or alternative algorithms for other resource monotones. Additionally, an interesting non-stabilizerness/entanglement relationship is shown to be present in the algorithm.

\section{Framework}
Let $\mathcal{H}_{n}=(\mathbb{C}^{2})^{\otimes n}$ be the $n$-qubit Hilbert space of dimension $d=2^n$, and $\mathcal{P}_{n}$ be the set of $n$-qubit Pauli strings --- these are all the $n$-fold tensor product of the qubit Pauli operators $\mathbb{I}, X, Y, Z$, where $\vert \mathcal{P}_n \vert = 4^{n}$. The set of unitary operators that map Pauli strings to Pauli strings are the Clifford unitaries, $\mathcal{C}_n$. Formally, these are the normaliser of the Pauli group. The $n$-qubit pure stabilizer states are then those states that can be generated from $ \ket{0}^{\otimes n}$ and $\mathcal{C}_n$, \hbox{$STAB \coloneq \big\{ C \ket{0}^{\otimes n} : C \in \mathcal{C}_n \big\}$}. Equivalently, a state $\ket{\psi}$ is a stabilizer state if there exists a subgroup, \hbox{$S \subset \mathcal{P}_{n}$}, of $d$ mutually commuting Pauli-strings such that \hbox{$\bra{\psi} P \ket{\psi} = \pm 1 \forall~ P \in S$} and \hbox{$\bra{\psi} P \ket{\psi} = 0 ~\forall~P \in \mathcal{P}_n \setminus S$}. Mixed stabilizer states are then convex combinations of pure stabilizer states, \ben{although we will focus on pure states here, unless specified otherwise.} 
  
The Pauli strings, $\mathcal{P}_n = \{ P_j \}_{j=0}^{d^2-1}$, form an operator basis, and hence, all states $\rho \in D(\mathcal{H}_n)$, where $D(\mathcal{H}_n)$ is the set of density operators on $\mathcal{H}_n$, can be decomposed as 
\begin{equation}
\begin{aligned}
    \rho &= d^{-1} \sum_{j=0}^{d^2-1} \textrm{tr} \big[ \rho P_j \big] P_j.\\ \label{eq: pauli-decomposition}
\end{aligned}
\end{equation}
For all pure states, $\ket{\psi} \in \mathcal{H}_n$, a probability distribution --- termed the characteristic distribution --- can be defined using the coefficients of $\ket{\psi}$'s Pauli-decomposition,
\begin{equation}
\begin{aligned}
   \mathcal{D}(\ket{\psi}) &\coloneq \big\{ d^{-1} \bra{\psi} P_{j} \ket{\psi}^2 \big\}_{j=0}^{d^2-1}.
\end{aligned}
\end{equation}
This can be interpreted as a probability distribution due to 
\begin{equation}
    d^{-1} \bra{\psi} P_{j} \ket{\psi}^2 \geq 0~ ~ {\rm and} ~ ~ \sum_{j=0}^{d^2-1} d^{-1} \bra{\psi} P_{j} \ket{\psi}^2 = 1 ~ ~ \forall~\ket{\psi}.
\end{equation}
See Appendix~\ref{SupplementaryMaterialA} for more details. Given their definition, it can immediately be seen that all stabilizer states have a characteristic distribution that consists of $d$ elements that are $1/d$, and $d^2-d$ elements that are zero. 

The $\alpha$-SRE \cite{PhysRevLett.128.050402} are defined as
\begin{equation}
    \begin{split}
            M_{\alpha}(\ket{\psi}) &= (1-\alpha)^{-1} \ln \bigg( d^{-1} \sum_{j=0}^{d^2-1} \bra{\psi} P_j \ket{\psi}^{2\alpha} \bigg),\label{SimpleAlphaRényi StabiliserEntropy}
    \end{split}
\end{equation}
and are the $\alpha$-Rényi entropies of the characteristic distribution shifted by $-\ln{(d)}$. It can then be seen that if one knows 
\begin{equation}
    A_{\alpha}(\ket{\psi}) = d^{-1} \sum_{j=0}^{d^2-1} \bra{\psi} P_j \ket{\psi}^{2\alpha},
\end{equation}
then they can find $M_{\alpha}(\ket{\psi})$ using classical post-processing. We therefore focus on finding $A_{\alpha}(\ket{\psi})$. 

In resource-theoretic terms, the $\alpha$-SRE have been shown to be a good measure of non-stabilizerness: they are faithful, $M_\alpha (\ket{\psi}) = 0$ if and only if $\ket{\psi} \in STAB$, invariant under allowed operations, \hbox{$M_\alpha(C \ket{\psi}) = M_\alpha(\ket{\psi})~$if$ ~ C \in \mathcal{C}_n$}, and additive, \hbox{$M_{\alpha}(\ket{\psi} \otimes \ket{\omega}) = M_\alpha(\ket{\psi}) + M_\alpha(\ket{\omega})$} \cite{PhysRevLett.128.050402}. Moreover, the $\alpha$-SRE are monotonic under the allowed operations of the resource theory of non-stabilizerness for $\alpha \geq 2$ \cite{PhysRevA.110.L040403}.

\section{SRE From Purity}

The purity of a quantum state $\rho$ is defined as \ben{\hbox{$\mathrm{Pur}(\rho) \coloneq \textrm{tr}[\rho^2]$}} \cite{nielsen_chuang_2010}. Here, we show that there exists a quantum state whose purity encodes the quantity $A_{\alpha}(\ket{\psi})$. Initially, consider a register containing $\alpha$ copies of $\ket{\psi}$, $\ket{\psi}_{\Tilde{B}}^{\otimes \alpha}$ where $\Tilde{B} = B_1B_2 \ldots B_\alpha$ labels the spaces initially holding the copies of $\ket{\psi}$. The agent then applies the mixed unitary quantum channel $\Paulichannel(\cdot)$ to this register, defined as
\begin{equation}
    \begin{split}
            \Paulichannel \big( \dyad{\psi}^{\otimes \alpha}_{\Tilde{B}} \big) & \coloneq d^{-2} \sum_{j=0}^{d^2-1} P_j^{\otimes \alpha} \big( \dyad{\psi}^{\otimes \alpha}_{\Tilde{B}} \big) P_j^{\otimes \alpha} \\
            &= d^{-2} \sum_{j=0}^{d^2-1} \big( P_j \dyad{\psi} P_j \big)^{\otimes \alpha}_{\Tilde{B}}.
    \end{split}
\end{equation}
This channel applies each Pauli string in $\mathcal{P}_n$ with equal probability to all $\alpha$ copies of $\ket{\psi}$. From a direct calculation of the purity of this state it can be seen that 
\begin{equation}
    \begin{split}
        \textrm{tr} \big[ \Paulichannel \big( \dyad{\psi}^{\otimes \alpha} \big)^2 \big] &= d^{-4}\sum_{j,k=0}^{d^2-1} \textrm{tr}\big[ P_j \dyad{\psi} P_j P_k \dyad{\psi} P_k \big]^{\alpha} \\
        &= d^{-2} \sum^{d^2-1}_{l=0} \bra{\psi} P_l \ket{\psi}^{2 \alpha} \\
        &= d^{-1} A_{\alpha}(\ket{\psi}). \label{purityofTidleB}
    \end{split}
\end{equation}
See Appendix~\ref{SupplementaryMaterialB.1} for a full proof. Given that the purity of a state is a measurable property, via the swap test \cite{barenco1996stabilisationquantumcomputationssymmetrisation, PhysRevLett.87.167902} or randomised measurements \cite{Elben2023} for example, this opens up an avenue for directly measuring $A_{\alpha}(\ket{\psi})$ for integer values of $\alpha$. To achieve this, one must first be able to prepare the state $\Paulichannel(\dyad{\psi}^{\otimes \alpha})$.

The channel $\Paulichannel(\cdot)$ is reminiscent of the Pauli-twirling channel, 
\begin{equation}
    \mathcal{T}(\cdot) \ben{\coloneq} d^{-2} \sum_{j=0}^{d^2-1} P_j \big( \cdot \big) P_j,
\end{equation}
which is a state preparation channel of the maximally mixed state i.e., $\mathcal{T}(\cdot) = {\rm tr}[\cdot](\mathbb{I}/d) $. \ben{In fact, it can be seen that if $\alpha=1$, then $\Paulichannel(\cdot)$ {\em is}} the Pauli-twirling channel. 
Then, interestingly, for \ben{arbitrary $\alpha>1$, within each subspace of $\Tilde{B}$ that initially contained a copy of $\ket{\psi}$}, a Pauli-twirling channel is applied, 
\ben{\begin{equation}
    \begin{split}
        {\rm tr}_{\Tilde{B} \setminus B_i} \big[ \Paulichannel \big( &\dyad{\psi}^{\otimes \alpha}_{\Tilde{B}} \big) \big] = d^{-2} \sum_{j=0}^{d^2-1} {\rm tr}_{\Tilde{B} \setminus B_i} \big[ \big(P_j \dyad{\psi} P_j \big)^{\otimes \alpha}_{\Tilde{B}} \big] \\  
        &= d^{-2} \sum_{j=0}^{d^2-1} (\bra{\psi} P_j P_j \ket{\psi})^{(\alpha-1)} P_j \dyad{\psi}_{B_i} P_j \\
        &= \mathcal{T}(\dyad{\psi}_{B_i}) ~ \forall~i \in \{1, \alpha\}, \label{eq:local_pauli_twiling}
    \end{split}
\end{equation}}
\ben{where $\textrm{tr}_{\Tilde{B} \setminus B_i}[\cdot ]$ means trace over all $\Tilde{B}$ other than $B_i$, and we have used both the linearity of the partial trace along with the fact that \hbox{$\bra{\psi} P_j P_j \ket{\psi}=1~\forall~j$}, as all Pauli-strings square to the identity. Given the Pauli-twirling channel is a state preparation channel of the maximally mixed state, $\Paulichannel \big( \dyad{\psi}^{\otimes \alpha}_{\Tilde{B}} \big)$ is maximally mixed in each local subspace that initially contained a copy of $\ket{\psi}$ i.e., in each $B_i \in \Tilde{B}$. Hence, $\Paulichannel \big( \dyad{\psi}^{\otimes \alpha}_{\Tilde{B}} \big)$ looks locally free in each subspace that initially contained a copy of $\ket{\psi}$, as \hbox{$\mathbb{I}/d \in STAB$}}. 

An additional interesting feature of $\Paulichannel \big( \dyad{\psi}^{\otimes \alpha}_{\Tilde{B}} \big)$ is that if one instead traces out some subset of $\Tilde{B}$, $\Bar{B} \subset \{1,\alpha\}$, the remaining state is equivalent to applying the channel to $\alpha = \vert \Bar{B} \vert$ copies of $\ket{\psi}$. 

\section{An Algorithm for Estimating the SRE}

Here, an algorithm for measuring $A_{\alpha}(\ket{\psi})$, via its encoding into the purity of $\Paulichannel (\dyad{\psi}^{\otimes \alpha})$, for integers $\alpha>1$ is presented. It requires $\mathcal{O}\big( \alpha d^2 \epsilon^{-2} \big)$ copies of $\ket{\psi}$, where $\epsilon$ is the additive error, and uses $2n$ ancilla qubits.

Initially, $\Paulichannel(\dyad{\psi}^{\otimes \alpha})$ must be prepared. Firstly, a coherent method to prepare this state is presented, in which the channel $\Paulichannel(\cdot)$ is applied to a subspace containing $\alpha$ copies of $\ket{\psi}$ via unitary dynamics. Subsequently, incoherent methods for preparing the state are also discussed. At the culmination of the algorithm, the purity of the prepared state must be measured, which can be achieved by one of a variety of methods. Here, we focus on the swap test, which is first reviewed before the algorithm is detailed.  

\subsection{Swap Test}

Given two states $\sigma$ and $\rho$, the swap test \cite{barenco1996stabilisationquantumcomputationssymmetrisation, PhysRevLett.87.167902} is a quantum algorithm for which $\textrm{tr} [ \sigma \rho ]$ can be directly estimated to additive error $\tau$ using $\mathcal{O}(\tau^{-2})$ copies of the states. If $\sigma$ and $\rho$ are pure states, the swap test therefore allows the fidelity between the states to be measured. If, instead, two copies of a single state, $\rho$, are used, then the swap test can be used to measure the purity of the state. See Algorithm~\ref{alg:swap_test} for an overview of the swap test for measuring purity, where $cSWAP$ is a controlled swap that controls on the ancilla.   

\begin{algorithm}[h!]
\caption{Swap Test For Purity}\label{alg:swap_test}
 \hspace*{\algorithmicindent} \textbf{Input:}  $\rho \otimes \rho$ \\
 \hspace*{\algorithmicindent} \textbf{Output:} $\mathrm{Pur}(\rho)$ 
\begin{algorithmic}[1]
\State Append qubit ancilla, $\ket{0} \otimes \rho \otimes \rho$ 
\State Apply $H \otimes \mathbb{I} \otimes \mathbb{I}$ 
\State Apply $cSWAP$ 
\State Apply $H \otimes \mathbb{I} \otimes \mathbb{I}$ 
\State Measure $Z$ on the ancilla qubit. Output the expectation value 
\end{algorithmic}
\end{algorithm}

\subsection{Algorithm for Estimating $A_{\alpha}(\ket{\psi})$ via Purity}\label{section: algorithm details}

Firstly, a register of $2n$ ancilla qubits is appended to a register containing $\alpha$ copies of $\ket{\psi}$,
\begin{equation}
\begin{aligned}
    \ket{\psi_{\alpha}}_{A \Tilde{B}} &= \ket{0}^{\otimes 2n}_{A} \otimes \ket{\psi}^{\otimes \alpha}_{\Tilde{B}}. \\
\end{aligned}
\end{equation}
where the subscript $A$ represents the ancilla system. Hadamard gates are then applied to each ancilla qubit to put them in an equal superposition of all the computational basis states, denoted by the set $\{ \ket{j} \}_{j=0}^{d^2-1}$. Next, the unitary 
\begin{equation}
    cU_{\mathcal{P}} \coloneq \sum_{j=0}^{d^2-1} \dyad{j}_{A} \otimes P_{j, \Tilde{B}}^{\otimes \alpha},
\end{equation}
is applied. This acts each of the $n$-fold Pauli strings on each of the $\alpha$ copies of $\ket{\psi}$ conditioned on the state in the ancillary system. Now, tracing out the ancillary space gives 
\begin{equation}
\begin{aligned}
    \psi'_{\alpha, \Tilde{B}} &= \textrm{tr}_{A} \big[ cU_{\mathcal{P}} \big(H^{\otimes 2n} \otimes \mathbb{I} \big) \ket{\psi_{\alpha}}_{A\Tilde{B}} \big] \\
    &= \Paulichannel \big( \dyad{\psi}^{\otimes \alpha} \big). \label{eq: partial trace over coherent state}
\end{aligned}
\end{equation}
Hence, $d^{-1}A_{\alpha}(\ket{\psi})$ can be output by measuring the purity of the $\Tilde{B}$ subspace. This can be achieved via the swap test, to an additive error $\tau$, using $\mathcal{O}(\alpha \tau^{-2})$ copies of $\psi'_{\alpha, \Tilde{B}}$. To measure $A_{\alpha}(\ket{\psi})$ to additive error $\epsilon$ therefore requires $\mathcal{O}(\alpha d^2 \epsilon^{-2})$ copies of $\ket{\psi}$. If a specific failure probability of $\delta$ is also specified, this can be stated more specifically as $\lceil \alpha d^2 \epsilon^{-2} \delta^{-1} \rceil$ copies of $\ket{\psi}$. See Appendix~\ref{SupplementaryMaterialC} for more details on resource requirements. An overview of the algorithm is detailed in Algorithm~\ref{alg:fullAlgorithm} and in Fig.~\ref{cicruitDiagram} a circuit diagram of the state preparation section of the algorithm for $\alpha=2$ and qubit states is presented. 

Performing the swap test requires $2$ copies of $\psi'_{\alpha, \Tilde{B}}$. Given each $\psi'_{\alpha, \Tilde{B}}$ uses $\alpha$ copies of $\ket{\psi}$, measuring $A_{\alpha}(\ket{\psi})$ using the swap test therefore requires $2\alpha$ copies of $\ket{\psi}$ to simultaneous exists, i.e, a $2\alpha$ copy measurements must be made. A quantum computer of at least $2n(\alpha+1) + 1$ qubits is therefore necessary when including ancilla. Alternative purity estimation algorithms, such as randomised measurement, could reduce this to $\alpha$ copy measurements on a $n(\alpha+2)$ qubit device by measuring the purity whilst acting on a single copy of $\psi'_{\alpha, \Tilde{B}}$ \cite{Elben2023} .  

It is noted that $cU_{\mathcal{P}}$ has a depth that is exponential in the system size and contains large controlled operations, making it practically challenging to implement. Although, each of the controlled Pauli-string gates, \hbox{$\dyad{j} \otimes P_j^{\otimes \alpha}$}, in $cU_\mathcal{P}$ can be decomposed into a sequence of controlled single-Pauli gates that builds the controlled Pauli-string on each of the $\alpha$ copies of $\ket{\psi}$ one element at a time. However, whilst this simplifies the control gates, it comes at the expenses of even greater depth. Below, alternative state preparation methods that do not make use of $cU_{\mathcal{P}}$ are detailed.  
\begin{figure} 
\begin{center} \label{cicruitDiagram}
\begin{quantikz}
    \lstick{$\ket{\psi}$} &      & \gate{X} \gategroup[4,steps=3,style={dashed,rounded corners, fill=blue!5, inner xsep=2pt},background,label style={label position=below,anchor=north,yshift=-0.2cm}]{{\sc
$cU_{\mathcal{P}}$}}         &  \gate{Y}                 &  \gate{Z}             &  \\
    \lstick{$\ket{\psi}$} &      & \gate{X}\wire[u]{q}&  \gate{Y}\wire[u]{q}      &  \gate{Z}\wire[u]{q}  &  \\ 
    \lstick{$\ket{0}$}  & \gate{H}& \ctrl{1}\wire[u]{q}&  \ctrl[open]{1}\wire[u]{q}&  \ctrl{1}\wire[u]{q}  & \ground{}          \\ 
   \lstick{$\ket{0}$} & \gate{H}& \ctrl{0}           &  \ctrl{0}                 &  \ctrl[open]{0}       &  \ground{}   
\end{quantikz}
\caption{A circuit diagram for the coherent algorithm for measuring $A_{2}(\ket{\psi})$ when $\ket{\psi}$ is a single qubit. Initially, $2n$ ancillary qubits are appended and placed in an even superposition of the computational basis states. The gate $cU_{\mathcal{P}}$ is then applied which encodes $A_{2}(\ket{\psi})$ into both the purity of the resistor holding the $\alpha$ copies of $\ket{\psi}$ and the ancilla. One of a numerous number of purity measuring algorithms can then be used to extract $A_{2}(\ket{\psi})$.}
\end{center}
\end{figure}
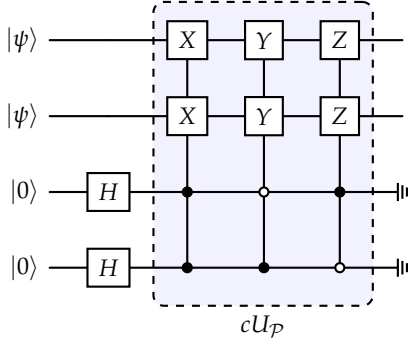

\subsection{Benchmarking}

The algorithm can be benchmarked for qubits using the state 
\begin{equation}
\ket{\psi_{\theta}} = \frac{1}{\sqrt{2}} \bigg( \ket{0} + e^{i \theta} \ket{1} \bigg),    
\end{equation}
given that $A_{\alpha}(\ket{\psi_\theta})$ can be simply written as a function of $\theta$ as
\begin{equation}
    A_{\alpha}(\ket{\psi_\theta}) = \frac{1}{2} \bigg(  1 + \cos(\theta)^{2\alpha} + \sin(\theta)^{2 \alpha} \bigg).
\end{equation}
These theoretical values are then compared to the output of a simulation performing Algorithm~\ref{alg:fullAlgorithm} for \hbox{$\alpha \in \{2,3,5,7\}$}. The results are plotted in Fig.~\ref{fig:varying_alpha} for additive error $\epsilon=0.05$ and failure probability $\delta= 0.1$ (see \cite{data} for the code and data). Error bars of $\epsilon$ are included and it can be seen that all points fall within $\epsilon$ of the theoretical value. Moreover, it can be seen that in  Fig.~\ref{fig:varying_alpha} that the accuracy of the algorithm appears to have no dependence on $\alpha$.  

\begin{figure}
    \centering
    \includegraphics[scale=0.185]{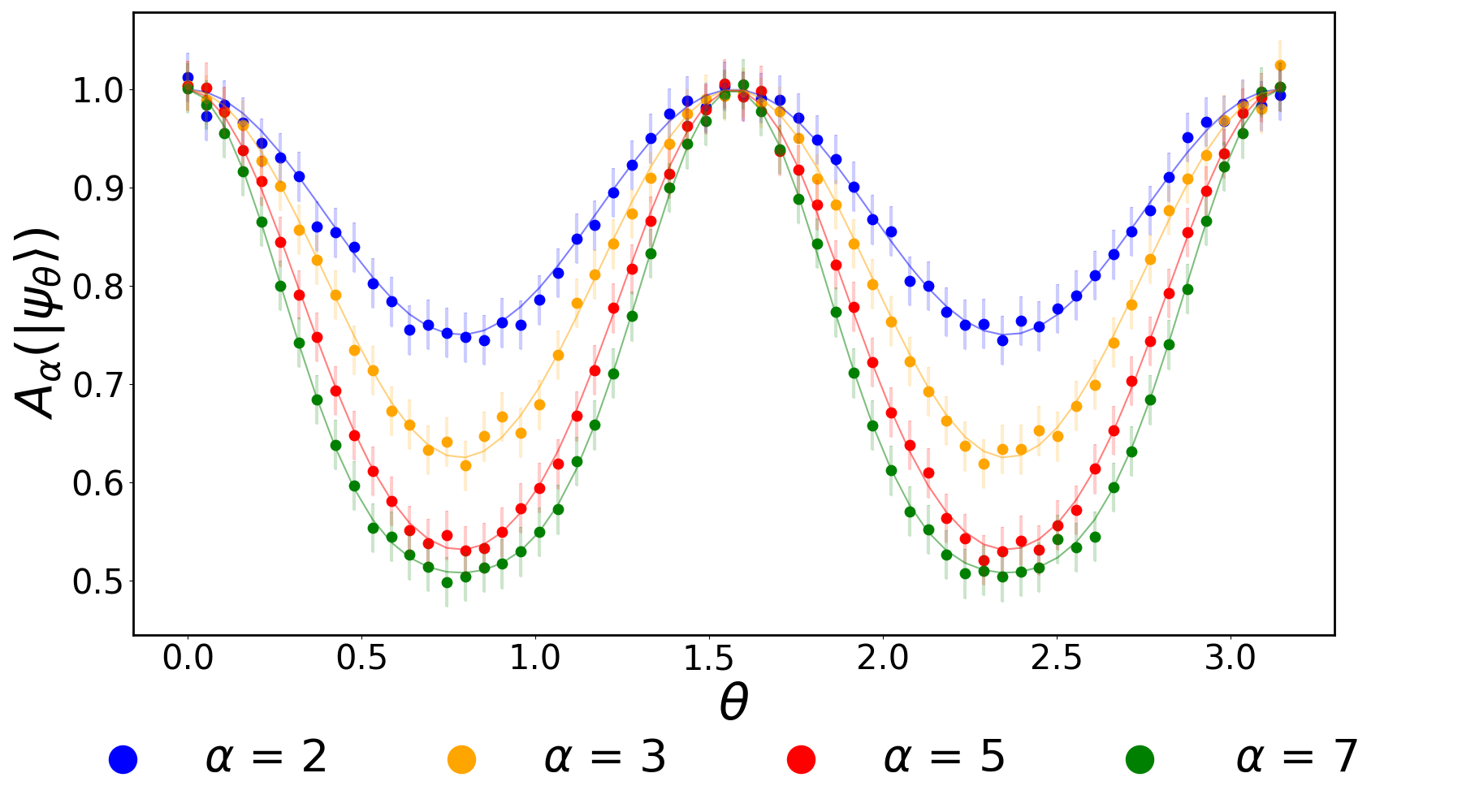}
    \caption{$A_{\alpha}(\ket{\psi_\theta})$ for $\alpha = \{2,3,5,7\}$ found using a circuit simulation (dots) of Algorithm~\ref{alg:fullAlgorithm} using $\lceil \alpha d^2 \epsilon^{-2} \delta^{-1} \rceil$ copies of $\ket{\psi_\theta}$, where $\epsilon=0.05$ and $\delta=0.1$, and the theoretical values (lines). Error bars of size $\epsilon$ are plotted.}
    \label{fig:varying_alpha}
\end{figure}


\subsection{Incoherent State Preparation Methods}
The channel $ \Paulichannel(\dyad{\psi}^{\otimes \alpha})$ is some mixed unitary channel; any mixed unitary channel can be performed by applying some unitary from a set with some corresponding probability, and then 'forgetting' which unitary was applied. The forgetting stage can be achieved by, for example, sending the state to another party without disclosing which unitary was applied, or by erasing any memory that contains information about the applied unitary. In the case of $ \Paulichannel(\dyad{\psi}^{\otimes \alpha})$, an $n$-qubit Pauli-string is chosen with probability $d^{-2}$ and then applied to all $\alpha$ copies of $\ket{\psi}$. Any number of methods for forgetting which Pauli-string was applied can then be performed, which then leaves $ \Paulichannel(\dyad{\psi}^{\otimes \alpha})$. This method therefore prepares the state $ \Paulichannel(\dyad{\psi}^{\otimes \alpha})$ without the need to use the expensive $cU_{\mathcal{P}}$ gate. 

\begin{algorithm}[h!]
\caption{Estimating $A_{\alpha}(\ket{\psi})$ via Purity}\label{alg:fullAlgorithm}
 \hspace*{\algorithmicindent} \textbf{Input:}  $\ket{\psi}^{\otimes \alpha}$ \\
 \hspace*{\algorithmicindent} \textbf{Output:} $A_{\alpha}(\ket{\psi})$ \\
\begin{algorithmic}[1]
\State Append $2n$ ancillary qubits, $\psi_\alpha = \ket{0}^{\otimes 2n} \otimes \ket{\psi}^{\otimes \alpha}$
\State Apply $H^{\otimes 2n} \otimes \mathbb{I}$ 
\State Apply $cU_{\mathcal{P}} = \sum_{i=0}^{d^2-1} \dyad{i} \otimes P_{i}^{\otimes \alpha}$
\State \ben{Measure $\mathrm{Pur}(\psi'_{\alpha, \Tilde{B}})$}, where 
\begin{equation}
    \begin{aligned}
        \psi'_{\alpha, \Tilde{B}} &= \textrm{tr}_{A} \big[ (cU_{\mathcal{P}})(H^{\otimes 2n} \otimes \mathbb{I}) \psi_{\alpha} \big], \\
    \end{aligned}
\end{equation} 
\State Output \ben{$d ~ \mathrm{Pur}(\psi'_{\alpha, \Tilde{B}})$}
\end{algorithmic}
\end{algorithm}   

\subsection{Comparisons to other algorithms}

There exists multiple methods one could implement to estimate $A_{\alpha}(\ket{\psi})$. Here, we will review and compare our algorithm to the leading algorithm, quantum state tomography (QST) and direct estimation. This is done via both the required number of copies of $\ket{\psi}$ \ben{(the copy complexity)} and the copy of the required measurement. Since all of the state preparation methods detailed above require an equal number of copies of $\ket{\psi}$, and all require only stabilizer operations to implement (which are considered to be free), the specific method used in preparation is unimportant when considering these two measures.  

We firstly compare to estimating $A_{\alpha}(\ket{\psi})$ via QST. To achieve this, it is first noted that $A_{\alpha}(\ket{\psi})$ can be written as the expectation value of an observable $\Gamma_{\alpha}^{\otimes n}$ acting on $2 \alpha$ copies of $\ket{\psi}$, 
\begin{equation}
    A_{\alpha}(\ket{\psi}) = \bra{\psi}^{\otimes 2\alpha} \Gamma_{\alpha}^{\otimes n} \ket{\psi}^{\otimes 2\alpha}, \label{replicaTrick}
\end{equation}
where $\Gamma_{\alpha} = 2^{-1} \sum^{3}_{i=0} (Q_i)^{\otimes 2 \alpha}$ where $Q_i$ are the single qubit Pauli-operators, $\mathbb{I}, X, Y, Z$ \cite{PhysRevB.107.035148}. Hence, one does not need to perform full state tomography on $\ket{\psi}$ to estimate $A_{\alpha}(\ket{\psi})$, they can instead aim to just estimate the expectation value of $\Gamma_{\alpha}^{\otimes n}$ directly. This has a sample complexity of $\Theta(d \vert \vert  \Gamma_{\alpha}^{\otimes n} \vert \vert_{\infty} \epsilon^{-2})$ when estimating $A_{\alpha}(\ket{\psi})$ to additive error $\epsilon$. Given $ \vert \vert \Gamma_{\alpha}^{\otimes n} \vert \vert_{\infty} = d$ if $\alpha$ is even and $1$ if it is odd \cite{PhysRevLett.132.240602}, the resource requirements of estimating $A_{\alpha}(\ket{\psi})$ using QST dependent on the parity of the $\alpha$ being considered. Specifically, for odd $\alpha$, one requires $\mathcal{O}(d \epsilon^{-2})$ copies of $\ket{\psi}$, whilst for even $\alpha$, $\mathcal{O}(d^3 \epsilon^{-2})$ copies are required. Hence, it can be seen that the number of copies of $\ket{\psi}$ needed for our algorithm scales worse than QST for odd $\alpha$, and better for even $\alpha$. See Appendix~\ref{SupplementaryMaterialC} for more details. 

Next, we consider measuring $A_{\alpha}(\ket{\psi})$ using direct estimation. For instance, one could perform independent calculations of the expectation values for each of the $d^2$ Pauli-strings in $\Gamma_{\alpha}^{\otimes n}$ with respect to the state $\ket{\psi}^{\otimes 2 \alpha}$. To estimate $A_{\alpha}(\ket{\psi})$ to additive error $\epsilon$, the expectation value of each individual Pauli-string in $\Gamma_{\alpha}^{\otimes n}$ would need to be measured to an additive error $\epsilon$ also. To do this, $\mathcal{O}(\alpha \epsilon^{-2})$ copies of $\ket{\psi}$ are needed for each string, meaning a total of $\mathcal{O}( \alpha d^2 \epsilon^{-2})$ copies of $\ket{\psi}$ are needed. Our algorithm therefore has the same resource requirements as performing direct estimation on $\Gamma_{\alpha}^{\otimes n}$. In addition, performing direct estimation on $\Gamma_{\alpha}^{\otimes n}$ would require $2 \alpha$ copy measurements to be made, which also matches our algorithm. One could perform direct estimation using single copy measurements by estimating the expectation values $\bra{\psi} P_j \ket{\psi}$ for all $P_i \in \mathcal{P}_n$, and then calculating $A_{\alpha}(\ket{\psi})$ via classical post-processing. However, this would require $\mathcal{O}(\alpha^2 d^4 \epsilon^{-2})$ copies of $\ket{\psi}$ to get an additive error of $\epsilon$, meaning an exponential increase in the resource requirements. See Appendix~\ref{SupplementaryMaterialC} for all details.

Lastly we consider the algorithm presented in \cite{PhysRevLett.132.240602} which \ben{also utilises Eq.~\eqref{replicaTrick}}. This algorithm is able to estimate $A_\alpha(\ket{\psi})$ to additive error $\epsilon$ using $\mathcal{O}(\alpha \epsilon^{-2})$ copies of $\ket{\psi}$ for odd $\alpha$, and \ben{$\mathcal{O}(\alpha d^2 \epsilon^{-2})$} for even $\alpha$. The algorithm measures two copies of $\ket{\psi}$ in the Bell basis, and then performs classical post-processing to extract $A_{\alpha}(\ket{\psi})$ from the measurement outcomes. 

\ben{For even $\alpha$, this algorithm scales better in terms of the number of copies of $\ket{\psi}$ needed than our algorithm. Moreover, the algorithm in \cite{PhysRevLett.132.240602} requires only two copy measurement, compared to the $2 \alpha$ copy measurements needed in our algorithm. Given this algorithm is therefore superior to ours in both these metrics, it represents the current optimal algorithm to measure the $\alpha$-SRE \ben{of an unknown state} for even $\alpha$. However, our algorithm does offer a benefit in terms of the number of qubits that must be measured for estimating the $\alpha$-SRE. If using the SWAP test to measure the purity at the culmination of our algorithm, only a single qubit needs to be measured. In the algorithm of \cite{PhysRevLett.132.240602}, $2n$ qubits must be measured. Therefore, in scenarios where measurement noise is a limiting factor, the reduced measurement overhead of our algorithm may make it preferable despite requiring more copies of $\ket{\psi}$. 

For odd $\alpha$, our algorithm scales equally to that of \cite{PhysRevLett.132.240602} in terms of number of copies required. In fact, it has been shown that measuring $A_{\alpha}(\ket{\psi})$ for odd $\alpha$ requires $\Omega(d^2 \epsilon^{-2})$ copies \cite{bittel2025operationalinterpretationstabilizerentropy}, meaning both algorithms are as efficient as possible for odd $\alpha$. Moreover, both algorithms demonstrate that this bound is tight.}

\ben{Lastly, we note that the depth of the circuit for performing our algorithm is greater than for the algorithm in \cite{PhysRevLett.132.240602}. Namely, their algorithm is of a fixed depth of two, with the circuit consisting of only Hadamard gates and CNOTs, used to turn a measurement in the computational basis into a measurement in the Bell basis. Ours, however, has a depth of at least $4^n$, making it more practically difficult to implement. Despite the depth, our circuit does contain only Clifford unitaries, meaning it is at least an allowed operation within the resource theory of non-stabilizerness.}

\ben{\subsection{Mixed states}

The $\alpha$-SRE is extend to mixed states \cite{PhysRevLett.128.050402, PhysRevA.107.022429} by taking the $\alpha$-Rényi entropies of the probability distribution
\begin{equation}
\begin{aligned}
   \mathcal{D}(\rho) &\coloneq \big\{ d^{-1} \mathrm{Pur}(\rho)^{-1} \textrm{tr}\big[ P_{j} \rho \big]^2 \big\}_{j=0}^{d^2-1},
\end{aligned}
\end{equation}
made from the Pauli-decomposition of $\rho$ given in Eq.~\ref{eq: pauli-decomposition}. Hence, the mixed $\alpha$-SRE is 
\begin{equation}
    \begin{split}
            M^{\mathrm{mix}}_{\alpha}(\rho) &= (1-\alpha)^{-1} \ln \bigg( d^{-1} \mathrm{Pur}(\rho)^{-1} \sum_{j=0}^{d^2-1} \textrm{tr}\big[ P_j \rho \big]^{2\alpha} \bigg). \label{SimpleAlphaRényi StabiliserEntropy}
    \end{split}
\end{equation}
As before, it can then be seen that if one knows 
\begin{equation}
    A^{\mathrm{mix}}_{\alpha}(\rho) = d^{-1} \mathrm{Pur}(\rho)^{-1}~\sum_{j=0}^{d^2-1} \textrm{tr}\big[ P_j \rho \big]^{2\alpha},
\end{equation}
then $M^{\mathrm{mix}}_{\alpha}(\rho)$ can be calculated via post-processing. However, if Algorithm~\ref{alg:fullAlgorithm} is performed on $\rho$, then the output is not $A^{\mathrm{mix}}_{\alpha}(\rho)$. This is specifically due to the use of \hbox{$\textrm{tr}\big[\rho^2\big] = \textrm{tr}\big[\rho \big]^2$} in the proof of Algorithm~\ref{alg:fullAlgorithm}, which holds if and only if $\rho$ is pure. Hence, we here note that the domain of Algorithm~\ref{alg:fullAlgorithm} is limited to pure states, but leave it to future work to determine if performing Algorithm~\ref{alg:fullAlgorithm} on $\rho$ could be used to place bounds on $ A^{\mathrm{mix}}_{\alpha}(\rho)$.} 

\section{Non-stabilizerness-Entanglement Correlation} 

Entanglement is considered to be a key ingredient of non-classicality. Non-stabilizerness is too, given that stabilizer operations are classically simulable. However, these two resources are known to be inequivalent as one can generate large amounts of entanglement via stabilizer operations, and large amounts of non-stabilizerness via entanglement non-generating operations. It could therefore be the case that these resources capture distinct notions of non-classicality. 

To explore this, it is essential to understand the interplay between these resources \cite{fattal2004entanglementstabilizerformalism, gu2024magic}. In this direction, we here demonstrate an example of a relationship between non-stabilizerness and entanglement in our algorithm, as measured by the $\alpha$-SRE and the Renyi entropy of entanglement respectively. The $\beta$-Renyi entropy of entanglement \cite{RevModPhys.81.865} of a bipartite pure state $\ket{\omega}_{AB}$, is defined as
\begin{equation}
    E_{\beta}(\ket{\omega}_{AB}) = (1-\beta)^{-1} \ln{\bigg( \sum_i \lambda_i^\beta \bigg)},
\end{equation}
where $\lambda_i$ are the Schmidt coefficients of $\omega_{AB}$. As with the SREs for non-stabilizerness, the Renyi entropies of entanglement are a good measure of entanglement within the resource theory of entanglement i.e., it is faithful and non-increasing under allowed operations (LOCC). 

\begin{theorem}\label{result:magic_entnaglment_trade_off}
    Let $\ket{\psi}$ be an $n$-qubit state, such that $d=2^n$, and $\alpha$ an integer where $\alpha \geq 1$. Then 
    \begin{equation}
        (1-\alpha) M_{\alpha}(\ket{\psi}) + E_{2} \big( \ket{\psi'_{\alpha}}_{A\Tilde{B}} \big) = \ln{(d)},
    \end{equation}
    where
    \begin{equation}
        \ket{\psi'_{\alpha}}_{A\Tilde{B}} = cU_{\mathcal{P}} \big(H^{\otimes 2n} \otimes \mathbb{I} \big) \ket{\psi_{\alpha}}_{A\Tilde{B}},
    \end{equation}
\end{theorem}

\begin{proof}
    All bipartite pure states have a Schmidt decomposition and hence 
    \begin{equation}
         \ket{\psi'_{\alpha}}_{A\Tilde{B}} = \sum_i \sqrt{\lambda_i} \ket{i}_A\ket{i}_{\Tilde{B}},
    \end{equation}
    where the $A \vert \Tilde{B}$ bipartition has been considered. The reduced state of $ \ket{\psi'_{\alpha}}_{A\Tilde{B}}$ in the $\Tilde{B}$ subspace,
    \begin{equation}
        \psi'_{\alpha, \Tilde{B}} = \textrm{tr}_{\Tilde{B}}\big[\dyad{\psi'_{\alpha}}_{A\Tilde{B}}\big],
    \end{equation}
    can also be expressed in terms of these Schmidt coefficients as
    \begin{equation}
        \psi'_{\alpha, \Tilde{B}} = \sum_{i} \lambda_i \dyad{i}_{\Tilde{B}}.
    \end{equation}
    The purity of $\psi'_{\alpha,\Tilde{B}}$ in terms of these Schmidt coefficients is then 
    \begin{equation}
        \begin{split}
            \textrm{tr} \big[\big(\psi'_{\alpha, \Tilde{B}}\big)^2 \big] &= \sum_{i} \lambda_i^2 \\
            &= d^{-1} A_{\alpha}(\ket{\psi}),
        \end{split}
    \end{equation}
    where the second line comes from Eq.~(\ref{purityofTidleB}). It can then be noted that 
    \begin{equation}
        \begin{split}
            E_{2}\big(\ket{\psi'_{\alpha}}_{A\Tilde{B}}\big) &= - \ln \bigg( \sum_i \lambda_i^2 \bigg) \\
            &= - \ln \big( d^{-1} A_{\alpha}(\ket{\psi}) \big) \\
            &= \ln(d) - (1-\alpha)M_{\alpha}(\ket{\psi}),
        \end{split}
    \end{equation}
    completing the proof. 
\end{proof}

As $\alpha \geq 1$, it can be seen that the more non-stabilizerness in the initial state, and hence the greater $M_{\alpha}(\ket{\psi})$, the more entanglement $cU_{\mathcal{P}}$ generates across the $A \vert \Tilde{B}$ bipartition. Moreover, as the reduced density operators of bipartite pure states are always the same, it can be seen that $A_{\alpha}(\ket{\psi})$ is also encoded into the purity of the other marginal,
\begin{equation}
\begin{aligned}
    \Tilde{\psi}_{\alpha, A} &= \textrm{tr}_{\Tilde{B}} \big[ cU_{\mathcal{P}} \big(H^{\otimes 2n} \otimes \mathbb{I} \big) \psi_{\alpha, A\Tilde{B}} \big] \\
    &= d^{-2} ~ \sum_{i,j=0}^{d^2-1}\textrm{tr} \big[ P_j P_i \dyad{\psi} \big]^{\alpha} \ket{i}\bra{j}_A,
\end{aligned}
\end{equation}
as it has the same coefficients of its reduced density operator. This can also be seen from a direct calculation, see Appendix~\ref{SupplementaryMaterialB.2} for a full derivation. This means that one can measure the purity of either the $\Tilde{B}$ or $A$ subspace of $\ket{\psi'_{\alpha}}_{A\Tilde{B}}$ to estimate $A_{\alpha}(\ket{\psi})$. Moreover, this demonstrates that $\Paulichannel(\dyad{\psi}^{\otimes \alpha})$ is not the unique state in which $A_{\alpha}(\ket{\psi})$ is encoded into the purity. It is left for future work to determine the full set of states for which $A_{\alpha}(\ket{\psi})$ is encoded into the purity. If a state could be found such that $A_{\alpha}(\ket{\psi})$ is encoded with a coefficient less than $d^{-1}$, it could lead to efficiency improvements over our algorithm. 

\ben{\section{Resource Hiding}

In this work, the state $\ket{\psi'_{\alpha}}_{A\Tilde{B}}$ has thus far been shown to possess two interesting features: measuring the purity of either the $A$ or $\Tilde{B}$ subsystem allows the $\alpha$-SRE to be estimated, and the entanglement across the $A\Tilde{B}$ bipartition is dependent on the $\alpha$-SRE of $\ket{\psi}$. Here, we report one final interesting feature of this state.  

It was shown in Eq.~\eqref{eq:local_pauli_twiling} that 
\begin{equation}
    {\rm tr}_{\Tilde{B} \setminus B_i} \big[ \Paulichannel \big( \dyad{\psi}^{\otimes \alpha}_{\Tilde{B}} \big) \big] = \mathbb{I}/d ~ \forall~i \in \{1, \alpha\}.
\end{equation}
As a result of both this and Eq.~\eqref{eq: partial trace over coherent state}, it can be concluded that 
\begin{equation}
    \textrm{tr}_{A\Tilde{B} \setminus B_i }\big[\dyad{\psi'_{\alpha}}_{A\Tilde{B}}\big] = \mathbb{I}/d ~ \forall~i \in \{1, \alpha\}, 
\end{equation}}
\ben{where $\textrm{tr}_{A \Tilde{B} \setminus B_i}\big[\cdot\big]$ means trace over all $A$ and $\Tilde{B}$ other than $B_i$. Hence, in each subspace of $\Tilde{B}$ in $\ket{\psi'_{\alpha}}_{A\Tilde{B}}$, the local state is maximally mixed. Therefore, in each subspace of $\Tilde{B}$ in $\ket{\psi'_{\alpha}}_{A\Tilde{B}}$ the state is a free state, as $\mathbb{I}/d \in \mathrm{STAB}$. Moreover, as $cU_{\mathcal{P}}^\dagger=cU_{\mathcal{P}}$, if $cU_{\mathcal{P}}$ is applied to $\ket{\psi'_{\alpha}}_{A\Tilde{B}}$, each subsystem within $\Tilde{B}$ returns to $\ket{\psi}$, with the local non-stabilizerness therefore restored. We note, as $cU_{\mathcal{P}}$ is Clifford, it is resource non-generating, meaning no non-stabilizerness is created under its application. Therefore, within $\ket{\psi'_{\alpha}}_{A\Tilde{B}}$, all the non-stabilizerness (resource) of the $\alpha$ copies of $\ket{\psi}$ has been reversible stored in the global correlations of the state, with none remaining locally in the subsystems of $\Tilde{B}$ that originally contained the copies of $\ket{\psi}$. }

\ben{This behaviour can be viewed as an instance of {\em resource hiding} \cite{985948}, with the non-stabilizerness of the $\alpha$ input states having become locally inaccessible in each subsystem of $\Tilde{B}$. Yet, the total non-stabilizerness remains unchanged, with it instead being stored in the global correlations. The resource can only be recovered --- and hence only used --- through the collaboration of all involved subsystems i.e., via a global operation. }

\ben{Interestingly, there are other states that also exhibit resource hiding under the action of $cU_{\mathcal{P}}$. For example, when $2n$ ancillary qubits in the maximally mixed state are appended to $\alpha$ copies of $\ket{\psi}$ i.e., $\mathbb{I}/(4^n) \otimes \dyad{\psi}^{\otimes \alpha}$ if $\ket{\psi}$ contains $n$ qubits. Moreover, we note that this is also an instance of resource hiding for the resource theory of purity, as this framework also has the maximally mixed state as a free state, and has all unitaries as allowed operations \cite{Streltsov_2018, resourceTheoryOfInformationalNonEquilibrium}. We leave further characterisations of this notion of resource hiding to future work.}

\section{Discussion}

In this work, we have demonstrated that the $\alpha$-SRE of a state $\ket{\psi}$, for positive integer values of $\alpha$, is encoded into the purity of the state $\mathcal{E}^{\alpha}_{\mathcal{P}}(\dyad{\psi}^{\otimes \alpha})$. We detail several methods for implementing this channel and, using known purity-estimation algorithms from the literature, use it to propose an alternative algorithm for estimating the $\alpha$-SRE of an unknown quantum state. When applying this channel via the unitary $cU_\mathcal{P}$, a positively correlated relationship between non-stabilizerness and entanglement was shown to be present. \ben{Lastly, the notion of resource hiding was discussed, where non-stabilizerness was stored non-locally}. 

\ben{Whilst the copy complexity of our algorithm outperforms quantum state tomography for even $\alpha$, and is optimal for odd $\alpha$ \cite{bittel2025operationalinterpretationstabilizerentropy}, there exists more efficient algorithms for estimating $A_{\alpha}(\ket{\psi})$ for even $\alpha$.} However, our algorithm introduces the idea of estimating resource monotones via a purity encoding. It would be interesting to see whether this method can be utilised to estimate other resource monotones, in both the resource theory of non-stabilizerness and beyond. Although, resource monotones often involve a minimisation over the set of free states. In fact, the absence of such a minimisation is partly what makes the SRE's powerful resource measures. It is unlikely to be possible to estimate resource measures with such a minimisation via purity. 

Moreover, within the coherent state preparation of $\mathcal{E}^{\alpha}(\dyad{\psi}^{\otimes \alpha})$, there was shown to be a relationship between two prominent resources, namely, non-stabilizerness and entanglement. Given the state $\ket{\psi'_{\alpha}}_{A\Tilde{B}}$, it was demonstrated that the entanglement across the $A \vert \Tilde{B}$ bipartition is related to the amount of non-stabilizerness in the initial state: the more non-stabilizerness in the initial state, the more entanglement across the bipartition, and vice-versa. Given the purity encoding of the resource measure this may seem surprising: the purity of a reduced state of a bipartite pure state is known to be inverse to the entanglement across that bipartition --- a pure state means there is no entanglement across the bipartition, a maximally mixed state means there is maximal entanglement across the bipartition. However, this apparent inconsistency arises due to the natural logarithm in the definition of $M_{\alpha}(\ket{\psi})$. If the purity of the reduced state is lower, the value of $A_{\alpha}(\ket{\psi})$ is therefore lower. Hence, $M_{\alpha}(\ket{\psi})$ is larger. Therefore, Result~\ref{result:magic_entnaglment_trade_off} is essentially the well-known purity entanglement trade-off present in bipartite pure states, but with the definition of $M_{\alpha}(\ket{\psi})$ meaning the non-stabilizerness and entanglement are instead positivity correlated. Interestingly, anytime a purity encoding of a resource monotone can be achieved by acting a unitary on pure states, that is, via a typical quantum circuit, a similar resource/entanglement trade-off is likely to be found. This could therefore represent a general method for finding relationships between general resources and entanglement. 

\begin{acknowledgments}

The author thanks Paul Skrzypczyk and Chung-Yun Hsieh for helpful discussions, Lorenzo Leone
for insightful comments and encouragement, and Jake Xuereb for useful comments on the manuscript. B.S. acknowledges support from UK EPSRC (EP/SO23607/1). 
\end{acknowledgments}

\bibliography{mainTextBib}

\onecolumngrid
\appendix

\section{Characteristic Distribution}  
\label{SupplementaryMaterialA}

For all pure states, $\rho = \dyad{\psi}$, the characteristic distribution can be defined as 
\begin{align}
    \mathcal{D}(\rho) &\coloneq \big\{ d^{-1} \bra{\psi} P_j \ket{\psi}^2 \},
\end{align}
where $\mathcal{P}_n = \{ P_j \}_{j=0}^{d^2-1}$. Firstly, note that $d^{-1} \bra{\psi} P_j \ket{\psi}^2$ is real as $P_j^\dagger = P_j$. Moreover, $ d^{-1} \bra{\psi} P_j \ket{\psi}^2 \geq 0~\forall~j$ as the square removes any negative values and $d > 0$. This distribution can therefore be interpreted as a probability distribution given that the sum of the elements is equal to one, which can be seen as
\begin{equation}
    \begin{split}
        d^{-1} \sum_{j=0}^{d^2-1} \bra{\psi} P_j \ket{\psi}^{2} &= d^{-1} ~ \sum_{j=0}^{d^2-1} \bra{\psi} P_j \dyad{\psi} P_j \ket{\psi} \\
        &= d^{-1} d^{2} ~ \bra{\psi} \bigg( d^{-2} \sum_{j=0}^{d^2-1} P_j \dyad{\psi} P_j \bigg) \ket{\psi} \\
        &= d \bra{\psi} \mathcal{T} (\dyad{\psi}) \ket{\psi} \\
        &= d \bra{\psi} \frac{\mathbb{I}}{d} \ket{\psi} \\
        &= 1,
    \end{split}
\end{equation}
where 
\begin{equation}
    \mathcal{T}(\cdot) = d^{-2} \sum_{j=0}^{d^2-1} P_j \big( \cdot \big) P_j
\end{equation} 
is the Pauli-twirling channel, which acts as a state-preparation channel of the maximally mixed state. 

\section{Stabilizer Rényi Entropy From Purity}\label{SupplementaryMaterialB}

This appendix contain proofs of how $A_{\alpha}(\ket{\psi})$ can be encoded into the purity of various quantum states.  

\subsection{$A_{\alpha}(\ket{\psi})$ from the purity of $\mathcal{E}_{\mathcal{P}}(\dyad{\psi}^{\otimes \alpha})$'s output} \label{SupplementaryMaterialB.1}

Here, we provide the proof of how $A_{\alpha}(\ket{\psi})$ is encoded into the purity of the state
\begin{equation}
    \begin{split}
            \mathcal{E}_{\mathcal{P}} \big( \dyad{\psi}^{\otimes \alpha} \big) &= d^{-2} \sum_{j=0}^{d^2-1} \big( P_j \dyad{\psi} P_j \big)^{\otimes \alpha}.
    \end{split}
\end{equation}
The purity of this state is given by
\begin{equation}
    \begin{split}
        \textrm{tr} \big[ \mathcal{E}_{\mathcal{P}} \big( \dyad{\psi}^{\otimes \alpha} \big)^2 \big] &= d^{-4}\sum_{j,k=0}^{d^2-1} \textrm{tr}\big[ P_j \dyad{\psi} P_j P_k \dyad{\psi} P_k \big]^{\alpha} \\
        &= d^{-4} \sum_{j,k=0}^{d^2-1} \big( \bra{\psi} P_j P_k \dyad{\psi} P_k P_j \ket{\psi} \big)^{\alpha}. \label{ap:mid_a_in_e}
    \end{split}
\end{equation}
It is then noted that
\begin{equation}
    P_j P_k = (i)^{\theta(j,k)}P_{l(j,k)}, ~~ P_kP_j = (-i)^{\theta(k,j)}P_{l(k,j)}, \label{PauliCommuting}
\end{equation}
where $P_{l(j,k)}=P_{l(k,j)}$ is a Pauli-string that depends on the input Pauli-strings, $P_j$ and $P_k$, and $\theta(j,k)=\theta(k,j)$ is a number that depends on a) the number and location of identities in the Pauli-strings $P_j$ and $P_k$, and b) the number of qubits for which the Pauli-strings apply the same Pauli operator. Using Eq.~(\ref{PauliCommuting}) allows the output Pauli-string from the multiplication to be separated from the resulting phase factor. 

Returning to Eq.~\eqref{ap:mid_a_in_e} and inputting Eq.~\eqref{PauliCommuting}, it can be seen that
\begin{equation}
    \begin{split}
        \textrm{tr} \big[ \mathcal{E}_{\mathcal{P}} \big( \dyad{\psi}^{\otimes \alpha} \big)^2 \big] &= d^{-4} \sum_{j,k=0}^{d^2-1} \big( \bra{\psi} (i)^{\theta(j,k)}P_{l(j,k)} \dyad{\psi} (-i)^{\theta(j,k)}P_{l(j,k)} \ket{\psi} \big)^{\alpha} \\
        &= d^{-4} \sum_{j,k=0}^{d^2-1} (i \times -i)^{\theta(j,k) \alpha} \big( \bra{\psi} P_{l(j,k)} \dyad{\psi}P_{l(j,k)} \ket{\psi} \big)^{\alpha} \\
        &= d^{-4} \sum_{j,k=0}^{d^2-1} \bra{\psi} P_{l(j,k)} \ket{\psi} ^{2\alpha}, \label{ap:mid_a_in_e_2}
    \end{split}
\end{equation}
meaning any phase term is irrelevant. By now rewriting $P_{l(j,k)} = P_jP_k$ (and ignoring the phase term as it has already been accounted for), it can be seen that 
\begin{equation}
    \begin{aligned}
        \sum_{j=0}^{d^2-1} P_{l(j,k)} &= \sum_{j=0}^{d^2-1}  P_jP_k \\
                                    &= \sum_{j'=0}^{d^2-1}  P_{j'}~, \label{eq:sum_over_pauli}
    \end{aligned}
\end{equation}
as multiplying each element of the set of Pauli-strings by a given Pauli-string ($P_k$ in the above) returns the set of Pauli-strings (when phase terms are ignored). Eq.~\eqref{ap:mid_a_in_e_2} therefore becomes 
\begin{equation}
    \begin{split}
      \textrm{tr} \big[ \mathcal{E}_{\mathcal{P}} \big( \dyad{\psi}^{\otimes \alpha} \big)^2 \big] &= d^{-4} d^2  \sum_{l=0}^{d^2-1} \bra{\psi} P_l \ket{\psi}^{2 \alpha} \\
      &= d^{-1} A_{\alpha}(\ket{\psi}),
    \end{split}
\end{equation}
completing the proof. 

\subsection{$A_{\alpha}(\ket{\psi})$ from Purity of Ancilla's}\label{SupplementaryMaterialB.2}

Here, we provide the proof of how $A_{\alpha}(\ket{\psi})$ is encoded into the purity of the ancillary qubits in Alg.~\ref{alg:fullAlgorithm} by direct calculation. 

Initially, $\alpha$ copies of $ \ket{\psi} $ are prepared in a register, before appending $2n$ ancillary qubits prepared in the computational basis state, leaving the complete register in the state
\begin{equation}
\begin{aligned}
    \ket{\psi_{\alpha}}_{A \Tilde{B}} &= \ket{0}^{\otimes 2n}_A \otimes \ket{\psi}^{\otimes \alpha}_{\Tilde{B}}, \\
\end{aligned}
\end{equation}
where $A$ labels the ancillary register whilst $\Tilde{B} = B_0B_1 \ldots B_\alpha$ labels the registers holding the $\alpha$ copies of $\ket{\psi}$. 

Hadamard gates are then applied to the ancillary qubits to put them in an equal superposition of the computational basis states, denoted by $\{ \ket{j} \}_{j=0}^{d^2-1}$, giving the following state,
\begin{equation}
    \begin{aligned}
        \ket{\Tilde{\psi}_{\alpha}}_{A \Tilde{B}} &= \big( H^{\otimes 2n}_A \otimes \mathbb{I}_{\Tilde{B}} \big) \ket{\psi_{\alpha}}_{A \Tilde{B}}  \\
        &= d^{-1}  \sum_{j=0}^{d^2-1} \ket{j}_A \otimes \ket{\psi}^{\otimes \alpha}_{\Tilde{B}}.
    \end{aligned}
\end{equation}
Next, the unitary is applied, 
\begin{equation}
    cU_{\mathcal{P}} = \sum_{i=0}^{d^2-1} \dyad{i} \otimes P_i^{\otimes \alpha},
\end{equation}
where $P_i~\in~\mathcal{P}_{n}~\forall~i$. The unitary $ cU_{\mathcal{P}}$ applies the Pauli-string $P_i$ on each copy of $\ket{\psi}$ if the ancillary qubits are in the state $\ket{i}$. Applying $ cU_{\mathcal{P}}$ gives, 
\begin{equation}
    \begin{aligned}
        \ket{\psi'_{\alpha}}_{A \Tilde{B}} &= cU_{\mathcal{P}} \ket{\Tilde{\psi}_{\alpha}}_{A \Tilde{B}} \\
        &= \bigg( \sum_{i=0}^{d^2-1} \dyad{i} \otimes P_i^{\otimes \alpha} \bigg) \biggl(d^{-1} \sum_{j=0}^{d^2-1} \ket{j}_A \otimes \psi^{\otimes \alpha}_{\Tilde{B}} \biggl) \\
        &= d^{-1} \sum_{i=0}^{d^2-1} \ket{i}_A \otimes (P_i \ket{\psi})_{\Tilde{B}}^{\otimes \alpha}.
    \end{aligned}
\end{equation}
If system $A$ is traced out, one gets 
\begin{equation}
\begin{aligned}
    \Tilde{\psi}_{\alpha, \Tilde{B}} &= \textrm{tr}_{A} \big[ cU_{\mathcal{P}} \big(H^{\otimes 2n} \otimes \mathbb{I} \big) \ket{\psi_{\alpha}}_{A\Tilde{B}} \big], \\
    &= \mathcal{E}^{\alpha}_{\mathcal{P}} \big( \dyad{\psi}^{\otimes \alpha} \big),
\end{aligned}
\end{equation}
meaning that $A_{\alpha}(\ket{\psi})$ is encoded into the purity the state in $\Tilde{B}$, as stated in the main text. If instead system $\Tilde{B}$ is traced out, one gets the following state  
\begin{equation}
    \begin{aligned}
        \psi^{'}_{\alpha, A}  &= \textrm{tr}_{\Tilde{B}} \bigg[ \ket{\psi'_{\alpha}} \bra{\psi'_{\alpha}}_{A \Tilde{B}}  \bigg]  \\ 
        &= d^{-2} \textrm{tr}_{\Tilde{B}} \Biggl[ \bigg(  \sum_{i=0}^{d^2-1} \ket{i}_A \otimes P_i \ket{\psi}^{\otimes \alpha}_{\Tilde{B}} \bigg) \bigg( \sum_{j=0}^{d^2-1} \bra{j}_A \otimes \bra{\psi}_{\Tilde{B}}^{\otimes \alpha} P_j^{\dagger} \bigg)  \Biggl] \\
        &= d^{-2} ~\textrm{tr}_{\Tilde{B}} \Biggl[ \sum_{i,j=0}^{d^2-1} \ket{i}\bra{j}_A \otimes (P_i \dyad{\psi} P_j)_{\Tilde{B}}^{\otimes \alpha} \Bigg]  \\
        &= d^{-2}~\sum_{i,j=0}^{d^2-1} \textrm{tr} \big[P_j P_i \dyad{\psi} \big]^{\alpha} \ket{i}\bra{j}_A,
    \end{aligned}
\end{equation}
where we have used the fact that the Pauli-strings are Hermitian. The purity of $\psi'_{\alpha, A}$ is then given by
\begin{equation}
    \begin{aligned}
        \textrm{tr} \big[ (\psi^{'}_{\alpha, A})^2 \big] &= d^{-4} \textrm{tr} \Biggl[ \biggl( ~\sum_{i,j=0}^{d^2-1} \textrm{tr} \big[P_j P_i \dyad{\psi} \big]^{\alpha} \ket{i}\bra{j}_A \biggl) \biggl( \sum_{k,l=0}^{d^2-1} \textrm{tr} \big[P_k P_l \dyad{\psi} \big]^{\alpha} \ket{l}\bra{k}_A \biggl)  \Biggl] \\
        &= d^{-4} ~  \sum_{i,j=0}^{d^2-1} ~ \sum_{k,l=0}^{d^2-1} \textrm{tr} \big[P_j P_i \dyad{\psi} \big]^{\alpha}  \textrm{tr} \big[P_k P_l \dyad{\psi} \big]^{\alpha} \delta_{jl} \delta_{ik} \\
        &= d^{-4} ~  \sum_{k,j=0}^{d^2-1}  \textrm{tr} \big[P_j P_k \dyad{\psi} \big]^{\alpha}  \textrm{tr} \big[P_k P_j \dyad{\psi} \big]^{\alpha} \label{swapTestOuput}
    \end{aligned}
\end{equation}

Here, using the same method as in Appendix~\ref{SupplementaryMaterialB.1} (inputting Eq.~\ref{PauliCommuting} and Eq.~\ref{eq:sum_over_pauli} into Eq.~\ref{swapTestOuput}), gives  
\begin{equation}
    \begin{aligned}
           \textrm{tr} \big[ (\psi^{'}_{\alpha, A})^2 \big] &=
        d^{-4} ~  \sum_{k,j=0}^{d^2-1}  \textrm{tr} \big[(i)^{\theta(j,k)}P_{l(j,k)} \dyad{\psi} \big]^{\alpha}  \textrm{tr} \big[(-i)^{\theta(j,k)}P_{l(j,k)} \dyad{\psi} \big]^{\alpha} \\
        &=  d^{-4} ~  \sum_{k,j=0}^{d^2-1}   (i \times -i)^{\theta(j,k) \times \alpha} ~ \textrm{tr}\big[P_{l(j,k)} \dyad{\psi} \big]^{\alpha}  \textrm{tr} \big[P_{l(j,k)} \dyad{\psi} \big]^{\alpha} \\
         &=  d^{-4} ~  \sum_{k,j=0}^{d^2-1}   \bra{\psi} P_{l(j,k)} \ket{\psi}^{2\alpha} \\
         &= d^{-1} A_{\alpha}(\ket{\psi}),
    \end{aligned}
\end{equation}
completing the proof. 

\section{Resource Requirements  \label{SupplementaryMaterialC}}

\subsection{Our Algorithm}

By repeating the swap test $\mathcal{O}(\tau^{-2})$ times, $\gamma = \textrm{tr} \big[ \Paulichannel(\dyad{\psi}^{\otimes \alpha})^2 \big]$ can be measured to an additive error of $\tau$. Each swap test uses two copies of $\Paulichannel(\dyad{\psi}^{\otimes \alpha})$, and each copy of $\Paulichannel(\dyad{\psi}^{\otimes \alpha})$ uses $\alpha$ copies of $ \ket{\psi} $. Hence, $\mathcal{O}( \alpha \tau^{-2} ) $ copies of $ \ket{\psi} $ are needed to measure $\gamma$ to additive error $\tau$, which can be shown using Chebyshev's inequality. To measure $A_{\alpha}(\ket{\psi})$ to additive error $\epsilon$ then requires $\epsilon = d \tau $, meaning $\mathcal{O}(\alpha d^2 \epsilon^{-2})$ copies of $\ket{\psi}$ are needed in total. 

\subsection{Quantum State Tomography}
Here, the calculation of the resource requirements for measuring $A_{\alpha}(\ket{\psi})$ using quantum state tomography (QST) are detailed. Let $\rho$ be a state and $O$ an observable such that one is aiming to find
\begin{equation}
    \expval{O}_{\rho} = \textrm{tr}\big[ \rho O \big].
\end{equation}
To do this via QST, one first constructs an estimate of the state, $\rho_{\rm est}$, and then calculates 
\begin{equation}
    \expval{O}_{\rm est} = \textrm{tr}\big[ \rho_{\rm est} O \big],
\end{equation}
via classical processing. The difference between these two values is then given by 
\begin{equation}
    \begin{split}
        \Delta \expval{O} &= \expval{O}_{\rho} - \expval{O}_{\rm est} \\
        &= \textrm{tr}\big[ \Delta \rho O \big],
    \end{split}
\end{equation}
where $\Delta \rho = \rho - \rho_{\rm est}$. It can then be seen that 
\begin{equation}
    \begin{split}
        \vert  ~ \Delta \expval{O} \vert \leq \vert \vert \Delta \rho \vert \vert_1 ~ \vert \vert O \vert \vert_{\infty},
    \end{split}
\end{equation}
by using Hölders inequality,
\begin{equation}
    \vert ~ \textrm{tr}\big[ A^\dagger B \big] ~ \vert \leq \vert \vert A \vert \vert_{p}~\vert \vert B \vert \vert_{q},
\end{equation}
where $p^{-1} + q^{-1} = 1$ and $p,q \geq 0$. Here we have considered $p=1$ and $q=\infty$. We now define
\begin{equation}
    \begin{split}
         \omega &\coloneq \vert \vert \Delta \rho \vert \vert_1, \\
         \epsilon &\coloneq \vert \vert \Delta \rho \vert \vert_1 ~ \vert \vert O \vert \vert_{\infty},
    \end{split}
\end{equation}
which are the trace norm error and the additive error respectively. The first quantifies the quality of the state that is reconstructed from performing QST; the second quantifies how far $\expval{O}_{\rm est}$ is from $\expval{O}$ and is the error considered in the main text. Now, to perform quantum state tomography to a trace norm error $\tau$ requires $\Theta(d\omega^{-2})$ copies of $\rho$. By noting that $\omega = \epsilon/ \vert \vert O \vert \vert_\infty$, it can be seen that in terms of the additive error, $\epsilon$, one requires $O(d \vert \vert O \vert \vert_\infty ^2 \epsilon^{-2})$ copies of $\rho$.

\subsection{Direct Estimation}

\subsubsection{Direct Estimation of $\Gamma_{\alpha}^{\otimes n}$}

Here, the aim is to estimate $\expval{\Gamma_{\alpha}^{\otimes n}}_{\psi^{2\alpha}} = \bra{\psi}^{\otimes 2 \alpha} \Gamma_{\alpha}^{\otimes n} \ket{\psi}^{\otimes 2\alpha}$ to additive error $\epsilon$ by individually estimating the expectation value of each of the $d$ Pauli-strings in $\Gamma_{\alpha}^{\otimes n}$, and then combining them via classical post-processing. 

For each Pauli-string $P_j$ in $ \Gamma_{\alpha}^{\otimes n}$, we assume $k$ copies of $\ket{\psi}^{\otimes 2 \alpha}$ are used to estimate \hbox{$\expval{P_j}_{\psi^{2\alpha}} = \bra{\psi}^{\otimes 2 \alpha} P_j \ket{\psi}^{\otimes 2\alpha}$}. Each copy is measured in a basis corresponding to $P_j$, giving the set of measurement outcomes $\{ x^j_l \}_{l=0}^{k-1}$. Each possible measurement has a standard deviation of $\sigma_{j}$. An unbiased estimator of $\expval{P_j}_{\psi^{2 \alpha}}$ is then given by 
\begin{equation}
    O_{P_j} = k^{-1} \sum^{k-1}_{l=0} x^j_l.
\end{equation}
This is repeated for all $d^2$ Pauli-strings in $\Gamma_{\alpha}^{\otimes n}$, giving an unbiased estimator of $\expval{\Gamma_{\alpha}^{\otimes n}}_{\psi^{2 \alpha}}$ as 
\begin{equation}
    O_{\Gamma} = d^{-1} k^{-1} \sum_{j=0}^{d^2-1} \sum_{l=0}^{k-1} x^j_l.
\end{equation}
The standard deviation of the complete set of measurements, $\sigma_{\Gamma}$, is then
\begin{equation}
    \begin{split}
        \sigma_{\Gamma}^2 &= d^{-2} k^{-2} \sum_{j=0}^{d^2-1} \sum_{l=0}^{k-1} \sigma_{j}^2 \\
        &= d^{-2} k^{-1} \sum_{j=0}^{d^2-1} \sigma_{j}^2. 
    \end{split}
\end{equation}
As $\sigma_j^2 \leq 1 ~\forall~j$ in $\Gamma_{\alpha}^{\otimes n}$, the above expression can be upper-bounded as $\sigma_{\Gamma}^2 \leq k^{-1}$. Using Chebyshev's inequality one can then see that 
\begin{equation}
    {\rm Pr} \bigg[ \big\vert  O_{\Gamma} - \expval{\Gamma_{\alpha}^{\otimes n}}_{\psi^{2 \alpha}} \big\vert \geq \epsilon \bigg] \leq \frac{\sigma_{\Gamma}^2}{\epsilon^2} \leq \frac{1}{k \epsilon^2}.
\end{equation}
Hence, one needs to use $\mathcal{O}(\epsilon^{-2})$ copies of $\ket{\psi}^{\otimes 2 \alpha}$ to estimate the expectation values of each $P_j$ in $\Gamma_{\alpha}^{\otimes n}$ such that $ O_{\Gamma}$ estimates $\expval{\Gamma_{\alpha}^{\otimes n}}_{\psi^\alpha}$ to additive error $\epsilon$. Given there are $d^2$ Pauli-strings in  $\Gamma_{\alpha}^{\otimes n}$, this gives a total of $\mathcal{O}(\alpha d^{2} \epsilon^{-2})$ copies of $\ket{\psi}$ needed. 

\subsubsection{Direct Estimation of Each Single-Copy Pauli-String}

Here, the aim is to estimate $A_{\alpha}(\ket{\psi})$ to additive error $\epsilon$ by individually estimating the expectation value of each Pauli-string $\expval{P_j}_{\psi} = \bra{\psi} P_j \ket{\psi}$ for all $P_j \in \mathcal{P}_n$, and then combining them via classical post-processing. This allows $A_{\alpha}(\ket{\psi})$ to be measured using single-copy measurements. 

Assume that the expectation value of each Pauli-string has been measured to an additive error of at most $\tau$, such that 
\begin{equation}
    \big\vert O_{P_j} - \expval{P_j}_{\psi} \big\vert \leq \tau ~ \forall ~ P_j \in \mathcal{P}_n,
\end{equation}
where $O_{P_j}$ is the estimate for the expectation value $ \expval{P_j}_{\psi}$. An estimate for $A_{\alpha}(\ket{\psi})$ can then be found as  
\begin{equation}
    \Tilde{A}_{\alpha} = d^{-1} \sum_{j=0}^{d^2-1} \big(O_{P_j} \big)^{2 \alpha}.
\end{equation}
To second order in $\tau$, the maximum a given $(O_{P_j})^{2 \alpha}$ differs from $ \expval{P_j}_{\psi}^{2 \alpha}$ is 
\begin{equation}
    \big\vert \Tilde{P}_j - \expval{P_j}_{\psi} \big\vert^{2 \alpha} \leq  2 \alpha (O_{P_j})^{2 \alpha -1} \tau + \mathcal{O}(\tau^2). 
\end{equation}
Therefore, again to second order in $\tau$, the maximum $\Tilde{A}_{\alpha}$ differs from $A_{\alpha}(\ket{\psi})$ is 
\begin{equation}
    \begin{split}
        \big\vert \Tilde{A}_{\alpha} - A_{\alpha}(\ket{\psi}) \big\vert &\leq \bigg\vert 2 \alpha d^{-1} \tau \sum_{j=0}^{d^2-1} \big( O_{P_j} \big)^{2 \alpha-1} \bigg\vert + \mathcal{O}(\tau^2) \\
        &\leq 2 \alpha d^{-1} \tau \sum_{j=0}^{d^2-1} \big\vert O_{P_j} \big\vert^{2 \alpha - 1} + \mathcal{O}(\tau^2) \\
        &\leq 2 \alpha d \tau , 
    \end{split}
\end{equation}
where the final line comes from the fact that $\vert O_{P_j} \vert \leq 1 ~\forall~P_j \in \mathcal{P}_n$. To ensure $A_{\alpha}(\ket{\psi})$ is measured to additive error $\epsilon$ one therefore needs $\tau \leq \epsilon/(2 \alpha d)$. 

With this, Chebyshev's inequality can be used to concluded that one needs $\mathcal{O}(\alpha^2 d^2 \epsilon^{-2})$ copies of $\ket{\psi}$ per Pauli-string to estimate $A_{\alpha}(\ket{\psi})$ to an addative error $\epsilon$. Hence, 
$\mathcal{O}(\alpha^2 d^4 \epsilon^{-2})$ copies of $\ket{\psi}$ are needed in total given there are $d^2$ Pauli-strings. 

\end{document}